\documentclass[12pt]{amsart}
\usepackage{latexsym,amsmath,amsfonts,amscd,amssymb}
\usepackage{graphicx}
\newtheorem{theorem}{Theorem}[section]
\newtheorem{theorem*}{Theorem}
\newtheorem{lemma}[theorem]{Lemma}
\newtheorem{proposition}[theorem]{Proposition}
\newtheorem{proposition*}[theorem*]{Proposition}
\newtheorem{corollary}[theorem]{Corollary}
\newtheorem{corollary*}[theorem*]{Corollary}

\newtheorem{definition}[theorem]{Definition}

\theoremstyle{remark}
\newtheorem{remark}[theorem]{Remark}
\newtheorem{remark*}[theorem*]{Remark}

\newtheorem{note*}[theorem*]{Note}

\newcommand{\EE}{{\mathbb E}}

\newcommand{\NN}{{\mathbb N}}

\title{On profitability of stubborn mining}

\subjclass[2010]{68M01, 60G40, 91A60.}
\keywords{Bitcoin, blockchain, proof-of-work, selfish mining, Catalan numbers}

\author[C. Grunspan]{Cyril Grunspan}
\address{Cyril Grunspan\newline{}\indent L\'eonard de Vinci P\^ole Univ, Research Center, Labex R\'efi\newline{}\indent Paris, France, }
\email{cyril.grunspan@devinci.fr}
\author[R. P\'{e}rez-Marco]{Ricardo P\'{e}rez-Marco}
\address{Ricardo P\'{e}rez-Marco\newline{}\indent CNRS, IMJ-PRG, Labex R\'efi \newline{}\indent Paris, France}
\email{ricardo.perez.marco@gmail.com}
\address{\tiny Author's Bitcoin Beer Address (ABBA)\footnote{\tiny Send some anonymous and moderate satoshis to support our research at the pub.}:\newline{}\indent 1KrqVxqQFyUY9WuWcR5EHGVvhCS841LPLn} 

\address{\includegraphics[scale=0.8]{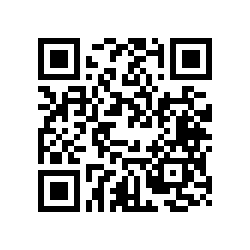}}

\begin{document}

\begin{abstract}
  We compute and compare profitabilities of stubborn mining strategies that are variations of 
  selfish mining. These are deviant mining strategies 
  violating Bitcoin's network protocol rules. We apply the foundational set-up from our previous companion 
  article on the profitability of
  selfish mining, and the new martingale techniques to get a closed-form computation for 
  the revenue ratio, which is the correct benchmark for profitability. Catalan numbers and Catalan distributions 
  appear in the closed-form computations. 
  This marks the first appearance of Catalan numbers in the Mathematics of the Bitcoin protocol.
\end{abstract}

{\maketitle}

\section{Introduction}

In our previous article \cite{GPM2018} we gave a rigorous foundation for the profitability analysis 
of alternative mining strategies in the 
Bitcoin network \cite{N08}. As for games with repetition, it depends on the proper 
analysis of the revenue and the duration over attack cycles. 
More precisely, we prove that the expected revenue $\EE[R]$ and expected duration $\EE[\tau]$ 
over an attack cycle give the ``Revenue Ratio''
$$
\Gamma=\frac{\EE[R]}{\EE[\tau]}
$$
This is the correct benchmark for the profitability of the strategy.

\medskip

This analysis was applied to the ``Selfish Mining'' strategy from \cite{ES14}. 
We also introduced in \cite{GPM2018} 
martingale tools that yield, by the application of Doob's Stopping Time Theorem, the expected 
duration of the attack cycles.
With the Markov model (used in \cite{ES14} and other articles in the literature) one cannot
compute the expected cycle duration $\EE[\tau]$. As we prove in \cite{GPM2018}, it is only after a 
difficulty adjustment that Selfish Mining or any other ``block withholding strategy'' can become profitable.
To compute the expected time of the next difficulty adjustment it is necessary to be able to compute the 
expected duration of cyles $\EE[\tau]$.

\medskip

In this article we apply again these new powerful techniques 
to some of the ``Stubborn Mining'' strategies presented 
in \cite{NKMS2016}. So far only a numerical Monte-Carlo analysis seems to be 
known for these strategies. 
For the ``Lead-Stubborn Mining'' (LSM) and the ``Equal Fork 
Stubborn Mining'' (EFSM) strategies we compute the ``Revenue Ratios''. 

\medskip

We fix some notations. Let $b>0$ be the block reward, and $\tau_0$ the average inter-block 
validation time for the total network (around $10$ minutes for the Bitcoin network). We denote 
by $q$ the relative hashing power of the attacker. Let 
$\gamma$ be the fraction of the honest network that the attacker attracts to mine on top of his fork. 
For a miner that after a difficulty adjustment has a Revenue Ratio $\tilde \Gamma$ we define 
his apparent hashrate $\tilde q$ by
$$
\tilde q= \frac{\tilde \Gamma \cdot \tau_0}{b} \ .
$$
The apparent hashrate of a miner can also be defined after a difficulty adjustment as the average
proportion of blocks mined by the miner in the official blockchain.
We make use of
$$
C(x) = \frac{1-\sqrt{1-4x}}{2x}=\frac{2}{1 + \sqrt{1 - 4 x}} =\sum_{n=0}^{+\infty} C_n x^n
$$
which is the generating series for Catalan numbers $(C_n)_{n\geq 0}$ (see Appendix \ref{appendix_Catalan_Distribution}).

\begin{theorem*}[Lead-Stubborn mining]
  The revenue ratio of the ``Lead-stubborn mining'' strategy is 
$$
\Gamma (\text{LSM}) = \left( q - \frac{pq (p - q)  (1 - \gamma)}{\gamma}
     \cdot \frac{1 - p (1 - \gamma) C ((1 - \gamma) pq)}{p + q (p - q)}
     \right)  \frac{b}{\tau_0} \ \ .
$$ 

After a difficulty adjustment, the apparent hashrate $\tilde{q}_{{LSM}}$ of the stubborn miner is
$$
\tilde{q}_{{LSM}} = q \cdot \frac{p + pq - q^2}{p + pq -q} - \frac{pq (p - q)  
(1 - \gamma)}{\gamma} \cdot \frac{1 - p (1 - \gamma) C ((1 - \gamma) pq)}{p + pq - q}  \ \ .
$$
\end{theorem*}

\begin{theorem*}[Equal Fork Stubborn mining]
The revenue ratio of the ``Equal Fork Stubborn mining'' strategy is 
$$
\Gamma (\text{EFSM}) = \left( q - \left( \frac{1 - \gamma}{\gamma}
     \right)  (p - q) \left ( 1 - pC ((1 - \gamma) pq)\right ) \right)  \frac{b}{\tau_0} \  \ .
$$
After a difficulty adjustment, 
the apparent hashrate $\tilde{q}_{{EFSM}}$ of the miner is
$$
\tilde{q}_{{EFSM}} = \frac{q}{p} - \frac{(1 - \gamma)  (p-q)}{\gamma p} \left (1 - pC ((1 - \gamma) pq)\right )  \ .
$$
\end{theorem*}

We compare these strategies to ``Honest Mining'' (HM)
and ``Selfish Mining'' (SM) and we determine in the  $(q, \gamma )$ parameter plane which one performs 
the best.

\section{Generalities.}

\subsection{Profitability.}
In \cite{GPM2018}, we studied the profitability of integrable repetition games which 
are composed of cycles, with a finite expected duration $\EE[\tau] <+\infty $. 
The stopping time $\tau$ is also called a \textit{strategy}. 
Mining strategies are repetition games, and sound mining strategies are integrable. 
For the comparison of profitability of two strategies we only need to compare the Revenue
Ratio of each one (see \cite{GPM2018})
$$
\Gamma =\frac{\EE[R]}{\EE[\tau]} 
$$
where $R$ is the revenue over a cycle.

The number of blocks $N'(t)$ and $N(t)$ validated by the attacker and honest miners 
respectively are independent Poisson processes
(see \cite{R12} for background on Poisson Process). 
The attack cycle ends when the honest miners catch-up the attackers. The number of validated blocks in a 
cycle in the official 
blockchain is $N (\tau) \vee N' (\tau)$. We denote by $T_1, T_2, \ldots$ (resp. $T'_1, T'_2, \ldots$) 
the inter-block validation time 
for the honest miners (resp. attackers).

\subsection{Profitability after a difficulty adjustment.}
The following Theorem describes how the Revenue Ratio changes after a difficulty adjustment (see \cite{GPM2018}).

\begin{theorem}\label{rap}
  After a difficulty adjustment the new Revenue Ratio $\tilde \Gamma$ is given by
  $$
  \tilde \Gamma = \Gamma \cdot \delta \ ,
  $$
  where
  $$
  \delta = \frac{\mathbb{E} [\tau]}{\tau_0 \cdot \mathbb{E} [N (\tau) \vee N' (\tau)]} \ .
  $$
\end{theorem}

\begin{proof}
After a difficulty adjustment, the expected revenue is the same, $\EE[\tilde R]=\EE[R]$, but the expected 
duration of a cycle is $\EE[\tilde \tau]= \tau_0 \cdot \mathbb{E} [N (\tau) \vee
  N' (\tau)]$. So, the new revenue ratio is $\tilde{\Gamma} (\tau) = \Gamma (\tau) \cdot
  \delta $.
\end{proof}

\subsection{Description of Stubborn strategies.}

We describe first the Selfish Mining (SM) strategy.
Let $\Delta \geq 0$ be the advance of the secret fork over the public blockchain. When the honest miners 
validate a block then the selfish miner does the following:

\begin{itemize} 
 \item If $\Delta=0$, he mines normally.
 \item If $\Delta = 1$ then he broadcasts his block. A competition follows.
 \item If $\Delta = 2$ then he broadcasts his secret fork.
 \item If $\Delta \geq 3$ then he broadcasts blocks from his secret fork to match the length of the public blockchain.
 \item Except in the first two cases, he keeps working on top of his secret fork.
\end{itemize}

\medskip

For the Lead-Stubborn Mining (LSM) strategy, with $\Delta \geq 2$ he proceeds as in 
the SM strategy for $\Delta \geq 3$ and with $\Delta =1$ he releases all his secret 
fork and mines normally on top of it.

In other words, a stubborn miner following the Lead Stuborn Mining strategy (LSM) waits until the 
honest miners catch up with him to broadcast all of his secret fork. Then, when this happens, there is a final round. 
Notice that a selfish miner following SM strategy does not take the risk of being caught by the honest
miners. If his advance shrinks to $1$, then he broadcasts his fork.

\medskip

For the Equal Fork Stubborn Mining (EFSM) strategy, everything is equal to LSM, 
but for $\Delta=1$ if he finds a new block
he does not reveal it. 

In other words a stubborn miner following
the Equal Fork Stubborn Mining strategy (EFSM) waits for the official
blockchain to overcome his secret fork by one block. He only gives up  when the
length of the official blockchain equals the length of his secret fork
plus one. In particular, the last round of the attack cycle  of the strategy is always lost by such
a miner. His reward comes only when blocks of
the official blockchain are built by honest miners on top of one of his
blocks. Indeed, the rogue miner never adds new blocks to the official
blockchain. He only tries to replace old blocks mined by the honest miners with
some of his blocks.


\section{Lead-Stubborn Mining strategy}

\subsection{Stopping time.} For each mining strategy we consider the stopping time 
associated with an attack cycle. 
Let $\tau_{LSM}$ be the stopping time of the ``Lead-Stubborn'' Mining strategy. 

\begin{proposition}
  We have 
  $$\tau_{LSM} = \tau + (T_{N (\tau) + 1} \wedge T'_{N (\tau) + 1}) \cdot {\bf{1}}_{T'_1 \leq T_1}
  $$ 
  with 
  $$
  \tau = \inf \{ t \geq T_1 ; N (t) = N' (t) +{\bf{1}}_{T_1 < T'_1} \} \ .
  $$
\end{proposition}

  In other words, either $T_1 < T'_1$ and the attack cycle ends at $T_1$ or
  $T'_1 \leq T_1$ and the attack cycle ends up a final round after the
  honest miner catch up with the attacker. 

  We calculate the expected duration time of an attack cycle.
  \begin{lemma} \label{exi}
  We have 
  $$
  \mathbb{E} [\tau] = \frac{p}{p - q} \,  \tau_0
  $$ 
  and 
  $$
  \mathbb{E} [\tau_{LSM}] = \mathbb{E}[\tau] + q \tau_0 
  = \tau_0 + 2 q \mathbb{E} [\tau] = \frac{p + pq - q^2}{p - q} \tau_0 \ .
  $$
  \end{lemma}

\begin{proof}
  By the strong Markov property (see \cite{R12}), we have
  \begin{align*}
    \mathbb{E} [\tau] & = \mathbb{E} [\tau |T_1 < T'_1] \cdot \mathbb{P}
    [T_1 < T'_1] +\mathbb{E} [\tau |T_1 > T'_1] \cdot \mathbb{P} [T_1 >
    T'_1]\\
    & = \mathbb{E} [T_1 |T_1 < T'_1] \cdot \mathbb{P} [T_1 < T'_1]
    +\mathbb{E} [T'_1 + \tilde{\tau} |T_1 > T'_1] \cdot \mathbb{P} [T_1 >
    T'_1]\\
    & = \mathbb{E} [T_1 \wedge T'_1] +\mathbb{E} [\tilde{\tau}] \cdot q
  \end{align*}
  where $\tilde{\tau} = \inf \{ t ; \tilde{N} (t) = \tilde{N}' (t) + 1
  \}$ with $\tilde{N} (t) = N (t + T'_1) - N (T'_1)$ and $\tilde{N}' (t) = N'
  (t + T'_1) - N' (T'_1)$. Both $\tilde{N}$ and $\tilde{N}'$ are Poisson
  processes with parameters $\alpha$ and $\alpha'$. Thus, from 
  Appendix 2 on Poisson Games, we have 
  $$
  \EE[\tilde \tau]= \frac{\tau_0}{p-q} 
  $$
  and
  $$
    \mathbb{E} [\tau] = \tau_0 + \frac{q}{p - q} \tau_0 = \frac{p}{p - q} \tau_0
  $$
  
  If $T_1 < T'_1$ then $\tau_{LSM} = \tau = T_1$.
  Otherwise (and this event occurs with probability $q$), once the honest
  miners catch-up with the attacker at $\tau$-time, there is a final round. 
\end{proof}

  Note that if $T_1 < T'_1$ then, $\tau_{LSM} = \tau = T_1$ and $N'(\tau) = 0$. 
  Otherwise, $T'_1 \leq T_1$ and $N' (\tau) > 0$.

  \subsection{Revenue Ratio.}

  For $n\geq 0$, we denote by $C_n$ the $n$-th Catalan number
 $$
 C_n = \frac{1}{2n+1} \binom{2n}{n} = \frac{(2 n) !}{n! (n + 1) !}  \ .
 $$ 
 We present in Appendix \ref{appendix_Catalan_Distribution} the combinatorial 
 properties of Catalan numbers used in this article, and the definition of $(p,q)$-Catalan distributions. 
 
The link between Catalan numbers and Bitcoin appears in the following lemma.

\begin{lemma} \label{cat}
The random variable $N'(\tau )$ follows the second type $(p,q)$-Catalan distribution, 
more precisely, we have $\mathbb{P}[N'(\tau)=0]=p$ and for $n \in \mathbb{N}^{\ast}$, 
$$
\mathbb{P} [N' (\tau) = n] = C_{n - 1}  (pq)^n \ .
$$
\end{lemma}

  \begin{proof} For $n=0$ we have $\mathbb P [N'(\tau)=0]=p$. Consider $n\geq 1$.
  The event $\{ N' (\tau) = n \}$ is the disjoint union of
  sub-events of the form $\{ \Sigma_1 < \ldots < \Sigma_{2 n + 1} \}$ where 
  for each $\Sigma_i$ there is $j$ such that 
  $\Sigma_i \in \{ S_j, S'_j \}$, and  $\Sigma_{2 n + 1} =
  S'_{n + 1}$. The sequence of points with coordinates $(N (\Sigma_i), N'
  (\Sigma_i))$ form a path starting at $(0, 0)$ and ending at $(n, n)$ which
  stays strictly above the first bisector $\{x = y\}$ in the Euclidean plane. For
  example, with $n = 3$,
  $$
  \{ N' (\tau) = 3 \} = \{ S'_1 < S'_2 < S'_3 < S_1 < S_2 < S_3 < S'_4 \}
     \cup \{ S'_1 < S'_2 < S_1 < S'_3 < S_2 < S_3 < S'_4 \} 
  $$
  The number of such paths is $C_{n - 1}$ (see Proposition \ref{geometric_property} in 
  Appendix \ref{appendix_Catalan_Distribution}). Moreover, the number of $S_j$
  (resp. $S'_j$) in the sequence of $(\Sigma_k)_{1 \leq k \leq 2 n}$
  is equal to $n$.
  \end{proof}

  From the expected value computation of a second type $(p,q)$-Catalan random variable 
  (Proposition \ref{catalan_expected_value} in Appendix \ref{appendix_Catalan_Distribution}) we get:
  
  \begin{corollary}\label{dercor}
  We have 
  $$
  \mathbb{E} [N'(\tau)] = \frac{p q}{p-q} \ .
  $$
  \end{corollary}

\begin{remark}
An alternative probabilistic proof follows from Doob's Theorem: We have 
$\mathbb{E} [N'(\tau \wedge t)] = \alpha' \mathbb{E} [\tau \wedge t]$ for $t>0$ and we let 
$t\rightarrow +\infty$ (similar to Theorem \ref{poiga} in Appendix \ref{appendix_Poisson_games}).
\end{remark}

Next we compute the expected revenue per attack cycle.

\begin{proposition} \label{erxi} Let $R_{LSM}$ be the revenue over an attack cycle.
We have 
$$
\mathbb{E} [R_{LSM}] = \left( \frac{p}{p- q} + q \right) qb - f (\gamma) b
$$ 
with
$$
  f (\gamma) = \frac{pq (1 - \gamma)}{\gamma} \cdot (1 - p (1 - \gamma) C ((1 - \gamma) pq))  \ .
$$
\end{proposition}

\begin{proof}
  Note that for $n>0$,
  \begin{eqnarray*}
    \mathbb{P} [R_{LSM} = nb|N' (\tau) = n] & = & \gamma p\\
    \mathbb{P} [R_{LSM} = (n + 1) b|N' (\tau) = n] & = & q\\
    \mathbb{P} [R_{LSM} < nb|N' (\tau) = n] & = & (1 - \gamma) p
  \end{eqnarray*}
  Moreover, during an attack cycle, each time the honest miners find a block 
  (except for the first block which is mined on a common root), there is a probability $\gamma$ 
  that it is found by a miner mining on top of
  the attacker's fork. If $N'(\tau) = n$ and $R(\tau_{LSM}) < n b$, this can happen at most $n-1$ times 
  over an attack cycle. So, by Lemma \ref{col} from Appendix \ref{appendix_biased_coin_tossing},
  \begin{eqnarray*}
    \mathbb{E} [R_{LSM} | (N' (\tau) = n) \wedge (R_{LSM} < nb)] &
    = & \left( n - \frac{1 - (1 - \gamma)^n}{\gamma} \right) b
  \end{eqnarray*}
  Therefore, by conditioning on $\tau$ and using Lemma \ref{exi}, Lemma \ref{cat}
  and Corollary \ref{dercor}, we get 
  \begin{align*}
    &\frac{\mathbb{E} [R_{LSM}]}{b} = \sum_{n > 0}
    \mathbb{E} \left[ \frac{R_{LSM}}{b} \middle |N' (\tau) = n \right]
    \cdot \mathbb{P} [N' (\tau) = n]\\
    & = \sum_{n > 0} \left( \left( n - \frac{1 - (1 - \gamma)^n}{\gamma}
    \right) \cdot (1 - \gamma) p + n \gamma p + (n + 1) q \right) \cdot
    \mathbb{P} [N' (\tau) = n]\\
    & = \sum_{n > 0} \left( n + 1 - \frac{p}{\gamma} + \frac{(1 -
    \gamma)^{n + 1}}{\gamma} p \right) \cdot \mathbb{P} [N' (\tau) = n]\\
    & = \mathbb{E} [N' (\tau)] + \left( 1 - \frac{p}{\gamma} \right) 
    (1 -\mathbb{P} [N' (\tau) = 0]) + \frac{p^2 q (1 - \gamma)^2}{\gamma} 
    \sum_{n > 0} C_{n - 1}  ((1 - \gamma) pq)^{n - 1}\\
    & = \frac{pq}{p - q} + \left( 1 - \frac{p}{\gamma} \right) q +
    \frac{p^2 q (1 - \gamma)^2}{\gamma} C ((1 - \gamma) pq)\\
    & = \left( \frac{p}{p - q} + q \right) q - \frac{pq}{\gamma} \cdot (1 -
    \gamma - p (1 - \gamma)^2 C ((1 - \gamma) pq))
  \end{align*}
  
\end{proof}

\begin{theorem}
  The revenue ratio of the Lead Stubborn strategy is
  $$
  \Gamma (LSM) = \left( q - \frac{pq (p - q)  (1 - \gamma)}{\gamma}
     \cdot \frac{1 - p (1 - \gamma) C ((1 - \gamma) pq)}{p + q (p - q)}
     \right)  \frac{b}{\tau_0}  \ .
  $$  
\end{theorem}

\begin{proof}
  We have $\Gamma (LSM) = \mathbb{E} [R_{LSM}]/\mathbb{E} [\tau_{LSM}]$. 
  Use Lemma \ref{exi} and Proposition \ref{erxi}.  
\end{proof}

\subsection{Difficulty adjustment.} We compute now the revenue ratio and the apparent hashrate after 
a difficulty adjustment.

\begin{lemma}\label{nle}
  We have 
  $\mathbb{E} [N(\tau_{LSM}) \vee N'(\tau_{LSM})] = \frac{\mathbb{E}[\tau_{LSM}]}{2 \tau_0} + \frac{1}{2}$.
\end{lemma}

\begin{proof}
  For all $t > 0$, $\tau_{LSM} \wedge t$ is a bounded stopping time. So,
  by proceeding as in Appendix 2, Doob's theorem yields $\mathbb{E} [N
  (\tau_{LSM} \wedge t)] = \alpha \mathbb{E} [\tau_{LSM} \wedge
  t]$ and $\mathbb{E} [N'(\tau_{LSM} \wedge t)] = \alpha' \mathbb{E}
  [\tau_{LSM} \wedge t]$. Taking limits when $t \rightarrow 0$, the
  monotone convergence theorem yields $\mathbb{E} [N (\tau_{LSM})] =
  \alpha \mathbb{E} [\tau_{LSM}]$ 
  and $\mathbb{E}[N'(\tau_{LSM})] = \alpha' \mathbb{E} [\tau_{LSM}]$. Now, we
  observe that at the end of an attack cycle, we have necessarily $| N
  (\tau_{LSM}) - N' (\tau_{LSM}) | = 1$. So, $N (\tau_{LSM})
  \vee N' (\tau_{LSM}) = \frac{N (\tau_{LSM}) + N'
  (\tau_{LSM}) + 1}{2}$. Hence we get the result by taking expected values
  on both sides of the last equality.
\end{proof}

The following proposition is now a consequence of Theorem \ref{rap}.

\begin{proposition}\label{dadjlsm}
  The parameter $\delta_{LSM}$ updating the difficulty is 
 $$
    \delta_{LSM} = \frac{p + pq - q^2}{p + pq - q} > 1 \ .
 $$
\end{proposition}

\begin{proof}
  Using  \ Theorem \ref{rap}, Lemma \ref{exi} and Lemma \ref{nle}, we have
  \[ \delta_{LSM} = \frac{\frac{\mathbb{E}
     [\tau_{LSM}]}{\tau_0}}{\frac{1}{2}  \left( \frac{\mathbb{E}
     [\tau_{LSM}]}{\tau_0} + 1 \right)} = \frac{\frac{\mathbb{E}
     [\tau]}{\tau_0} + q}{1 + q \frac{\mathbb{E} [\tau]}{\tau_0}} =
     \frac{\frac{p}{p - q} + q}{1 + \frac{pq}{p - q}} = \frac{p + pq - q^2}{p
     + pq - q} > 1 \]
\end{proof}

\subsection{Apparent hashrate after a difficulty adjustment}

From Proposition \ref{dadjlsm}, we can deduce the hashrate of the strategy on the long term.
 \begin{corollary}
  After a difficulty adjustment, the apparent hashrate
  $\tilde{q}_{{LSM}}$ is 
  $$
    \tilde{q}_{{LSM}}  =  q \cdot \frac{p + pq - q^2}{p + pq -
    q} - \frac{pq (p - q)  (1 - \gamma)}{\gamma} \cdot \frac{1 - p (1 -
    \gamma) C ((1 - \gamma) pq)}{p + pq - q}
  $$
\end{corollary}

\begin{proof}
  By Theorem \ref{rap}, we have:
  \begin{align*}
    &\tilde{\Gamma} (LSM) = \Gamma (LSM) \delta_{{LSM}}\\
    & = \left( q - \frac{pq (p - q)  (1 - \gamma)}{\gamma} \cdot \frac{1 -
    p (1 - \gamma) C ((1 - \gamma) pq)}{p + q (p - q)} \right) 
    \frac{b}{\tau_0} \cdot \frac{p + pq - q^2}{p + pq - q}\\
    & = \left( q \cdot \frac{p + pq - q^2}{p + pq - q} - \frac{pq (p - q) 
    (1 - \gamma)}{\gamma} \cdot \frac{1 - p (1 - \gamma) C ((1 - \gamma)
    pq)}{p + pq - q} \right)  \frac{b}{\tau_0}
  \end{align*}
  
\end{proof}

\section{Equal Fork Stubborn Mining strategy}

\subsection{Stopping time.} Let $\tau_{EFSM}$ be the stopping time of an attack 
cycle for the Equal Fork Stubborn Mining strategy. 

\begin{proposition}
  We have  
  $$
  \tau_{EFSM} = \inf \{ t \geq 0 ; N (t) = N' (t) + 1 \} \ .
  $$
\end{proposition}

\begin{proof}
If $T_1 < T'_1$, then we have $\tau_{EFSM} = \tau = T_1$. 
Otherwise, we wait for the honest miners to catch up with the stubborn miner 
and win the last round. 
\end{proof}

From Theorem \ref{poiga} in Appendix \ref{appendix_Poisson_games} we get:
\begin{lemma} \label{pgame}
We have 
$$
\mathbb{E} [\tau_{EFSM}] = \frac{\tau_0}{p - q} \ .
$$
\end{lemma}

\subsection{Revenue ratio.}
  We denote by $R_{EFSM}$ the revenue of the stubborn miner after an attack cycle.
  \begin{lemma}
  The random variable $N'(\tau_{EFSM})$ is a $(p,q)$-Catalan distribution, i.e. 
  for $n \geq 0$, we have 
  $$
  \mathbb{P} [N' (\tau_{EFSM}) = n] = C_n \,  p (pq)^n \ .
  $$
  \end{lemma}

  \begin{proof}
  The event $\{ N' (\tau_{EFSM}) = n \}$ can be decomposed as a disjoint
  union of sub-events of the form $\{ \Sigma_1 < \ldots < \Sigma_{2 n + 1} <
  \Sigma_{2 n + 2} \}$ where for each $\Sigma_i$ there is $j$ such that $\Sigma_i \in \{ S_j, S'_j
  \}$, and $\Sigma_{2 n + 1} = S_{n + 1}$ and  $\Sigma_{2 n + 2} = S'_{n + 1}$. 
  The sequence of points with coordinates
  $(N (\Sigma_i), N' (\Sigma_i))$ for $i \in \{1,\ldots, 2 n + 1 \}$
  form a path starting at $(0, 0)$ and ending at $(n + 1, n)$ which never crosses
  the first diagonal $\{x = y\}$ in the Euclidean plane before reaching the point $(n + 1, n)$.
  For example,
  \[ \{ N' (\xi) = 2 \} = \{ S'_1 < S'_2 < S_1 < S_2 < S_3 < S'_3 \} \cup \{
     S'_1 < S_1 < S'_2 < S_2 < S_3 < S'_3 \} \]
  The number of such paths is $C_n$ (see Proposition \ref{geometric_property} from 
  Appendix \ref{appendix_Catalan_Distribution}). Moreover, the number of $S_j$ (resp.
  $S'_j$) in the sequence of $(\Sigma_k)_{1 \leq k \leq 2 n + 1}$ is
  equal to $n + 1$ (resp. $n$). Hence we get the result.
  \end{proof}

  From the expected value computation of a $(p,q)$-Catalan random variable 
  (Proposition \ref{catalan_expected_value} in Appendix \ref{appendix_Catalan_Distribution}) we get:
  
  \begin{corollary}\label{dercor2}
  We have 
  $$
  \mathbb{E} [N'(\tau_{EFSM})] = \frac{q}{p-q} \ .
  $$
  \end{corollary}

We compute now $\mathbb{E} [R_{EFSM}]$. 

\begin{proposition} \label{erxif}
We have 
$$
\mathbb{E} [R_{EFSM}] = \frac{q}{p - q} b - g (\gamma) b
$$ 
with
$$ 
g (\gamma) =  \frac{1 - \gamma}{\gamma}   \left (1 - pC ((1 -\gamma) pq)\right ) \ .
$$
\end{proposition}

\begin{proof}
  By definition of $\tau_{EFSM}$, if we know that $N' (\tau_{EFSM}) =
  n$, then the honest miners have mined $n+1$ blocks at time $t=\tau_{EFSM}$.
  Except for the first block, the probability that
  the block validated by the honest miners is mined on a fork created by the stubborn miner is $\gamma$.
  So, by Lemma \ref{col} from Appendix \ref{appendix_biased_coin_tossing},
  $$ 
  \mathbb{E} \left[ \frac{R_{EFSM}}{b} \middle | N'
     (\tau_{EFSM}) = n \right] = n + 1 - \frac{1 - (1 - \gamma)^{n +
     1}}{\gamma}  \ .
  $$
  Therefore, we have 
  \begin{align*}
    & \frac{\mathbb{E}[R_{EFSM}]}{b} = \sum_{n \geq 0}
    \mathbb{E} \left[ \frac{R_{EFSM}}{b} \middle |N' (\tau_{EFSM})
    = n \right] \cdot \mathbb{P} [N' (\tau_{EFSM}) = n]\\
    & = \sum_{n \geq 0} \left( n + 1 - \frac{1 - (1 - \gamma)^{n +
    1}}{\gamma} \right) \cdot \mathbb{P} [N' (\tau_{EFSM}) = n]\\
    & = \sum_{n \geq 0} n  \, \mathbb{P} [N' (\tau_{EFSM}) =
    n] +  1 - \frac{\gamma - 1}{\gamma}   \sum_{n \geq 0} \mathbb{P}
    [N' (\tau_{EFSM}) = n] 
    + \frac{1 - \gamma}{\gamma} p \sum_{n \geq 0} ((1 - \gamma)
    pq)^n C_n\\
    & = \mathbb{E} [N' (\tau_{EFSM})] -  \frac{1 -
    \gamma}{\gamma}  + \frac{1 - \gamma}{\gamma}  pC ((1
    - \gamma) pq)\\
    & = \frac{q}{p - q} - \frac{1 - \gamma}{\gamma}   \left ( 1 - pC
    ((1 - \gamma) pq)\right ) \ .
  \end{align*}

  \end{proof}

\begin{theorem}
  The revenue ratio of the ``Equal Fork Stubborn mining'' strategy is
  $$
  \Gamma (EFSM) = \left( q -  \frac{1 - \gamma}{\gamma}
      (p - q) \left ( 1 - pC ((1 - \gamma) pq)\right ) \right)  \frac{b}{\tau_0}  \ .
  $$
\end{theorem}

\begin{proof}
 Use  Lemma \ref{pgame} and Proposition \ref{erxif}.
\end{proof}

\subsection{Difficulty adjustment.}
\begin{proposition}
  \label{axif} The parameter $\delta_{EFSM}$ updating the difficulty is 
  $$
  \delta_{EFSM} = \frac{1}{p} > 1 \ .
  $$
\end{proposition}

\begin{proof}
  At the end of an attack cycle, the number of new blocks in the official
  blockchain is $N (\tau_{EFSM})$ and $\mathbb{E} [N
  (\tau_{EFSM})] = \alpha \mathbb{E} [\tau_{EFSM}]$. Therefore we have
  $$
  \delta_{EFSM} = \frac{\frac{\mathbb{E}
  [\tau_{EFSM}]}{\tau_0}}{\mathbb{E} [N (\tau_{EFSM})]} =
  \frac{1}{p} \ .
  $$
\end{proof}

\subsection{Apparent hashrate after a difficulty adjustment.}
 It's now easy to get the long term apparent hashrate of the EFSM strategy.
 \begin{corollary}
  After a difficulty adjustment, the apparent hashrate is 
  $$
  \tilde{q}_{EFSM} = \frac{q}{p} - \frac{(1 - \gamma) (p -q)}{\gamma p} \left (1 - p \, C ((1 - \gamma) pq)\right )  \ .
  $$
 \end{corollary}

\begin{proof}
  Use Theorem \ref{rap} and Proposition \ref{axif}.
\end{proof}

\section{Comparison of strategies.}

For different values of the parameters $q$ and $\gamma$, we can compare the profitability 
of the different strategies after a 
difficulty adjustment by comparing the revenue ratios after a difficulty adjustment, 
or, equivalently, their apparent hashrate. 


We consider the four strategies: 
\begin{itemize}
 \item Honest Mining (HM).
 \item Selfish Mining (SM).
 \item Lead-Stubborn Mining (LSM).
 \item Equal Fork Stubborn mining (EFSM).
\end{itemize}

We color the region $(q, \gamma) \in [0,0.5] \times [0,1]$ according to which strategy is 
more profitable, and we obtain Figure 1.

\begin{figure}[!ht]
   \includegraphics[height=7cm, width=9cm]{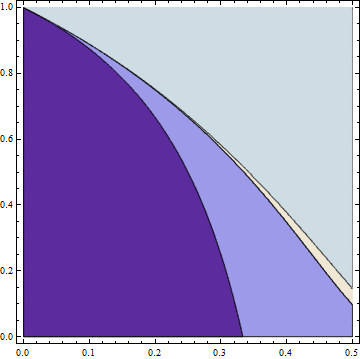}
   \caption{\label{fig}
   Profitability. From left to right: HM, SM, LSM, EFSM.
   \medskip}
\end{figure}

From left to right, the best strategy is successively HM, SM, LSM and EFSM. 
Note that, the LSM strategy is superior 
only on a thin domain.

\medskip

A similar picture is numerically computed by Montecarlo simulations in  \cite{NKMS2016} where 
only a numerical study is 
carried out. See Figure 2 (disregard regions R4 to R7 that correspond to ``Trail $T_j$ Stubborn Mining'' 
strategies that are studied in \cite{GPM2018-2}).

\begin{figure}[!ht]
   \includegraphics[height=8cm, width=10cm]{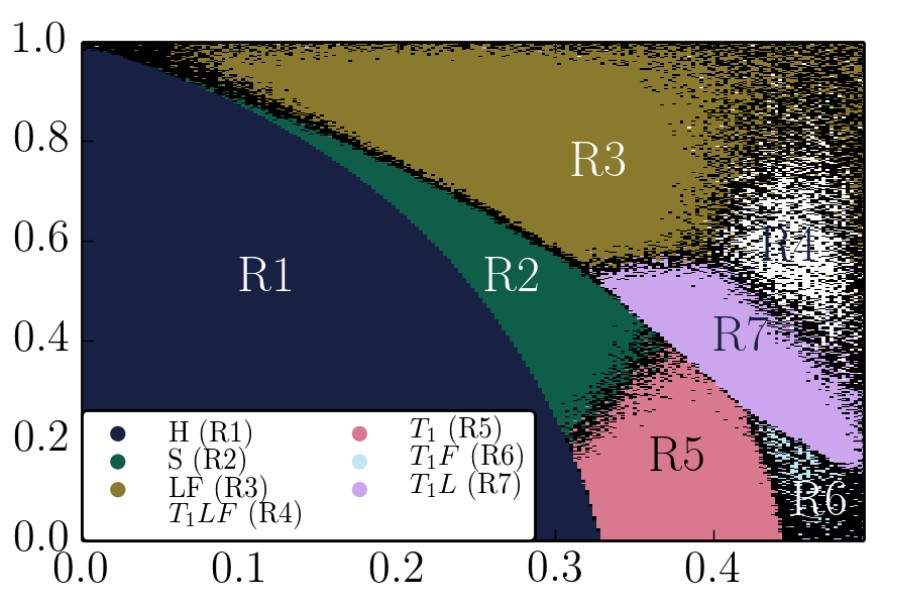}
   \caption{\label{fig2} Figure from \cite{NKMS2016}.}
\end{figure}

Regions R1, R2 and R3 correspond respectively to HM, SM and EFSM. We notice 
the absence of the region corresponding to LSM in between R2 and R3. 
Apparently the numerical methods did not detect it (also it is  noticeable 
how in their figure the boundary between R2 and R3 is more blurred than other boundaries,
as for example the one separating R1 and R2).

%

\appendix

\section{Catalan distributions.}\label{appendix_Catalan_Distribution}

For $n\geq 0$, we denote by $C_n$ the $n$-th Catalan number
 $$
 C_n = \frac{1}{2n+1} \binom{2n}{n} = \frac{(2 n) !}{n! (n + 1) !}  \ .
 $$ 
 We refer to \cite{K08} for background and combinatorial properties of Catalan numbers. 
 Their generating series is 
 $$
 C (x) = \sum_{n=0}^{+\infty} C_n x^n = \frac{1-\sqrt{1-4x}}{2x}=\frac{2}{1 + \sqrt{1 - 4 x}} \ .
 $$
Observe that $\sqrt{1-4pq}=p-q$ and $C(pq)=1/p$, which justifies the following definition:

 \begin{definition}[Catalan distributions]
 Let $1/2<p<1$ and $q=1-p$.
 A random variable $X$ taking values in $\NN$ is a $(p,q)$-Catalan random variable if it 
 follows the $(p,q)$-Catalan distribution, that is,  for $n\geq 0$
 $$
 \mathbb P [X=n]= C_{n} p(pq)^n  \ .
 $$
 The second type $(p,q)$-Catalan distribution is defined by $\mathbb P [X=0]=p$ and for $n\geq 1$,
 $$
 \mathbb P [X=n]= C_{n-1} (pq)^n  \ .
 $$
 \end{definition}

 \begin{proposition}\label{catalan_expected_value}
  The expected value of a  $(p,q)$-Catalan random variable $X$ is
  $$
  \EE[X]=\frac{q}{p-q}
  $$
  If $X$ is a second type $(p,q)$-Catalan random variable then
  $$
  \EE[X]=\frac{pq}{p-q} \ .
  $$
 \end{proposition}
\begin{lemma}
We have
$$
\frac{d}{dx} \left ( xC(x)\right ) \bigg\rvert_{x=pq} =\frac{1}{p-q} \ .
$$
\end{lemma}
\begin{proof}[Proof of the Lemma]
 We have $\frac{d}{dx} \left ( xC(x)\right )=(1-4x)^{-1/2}$ and the result follows.
\end{proof}

\begin{proof}[Proof of the Proposition]
For the $(p,q)$-Catalan random variable we have
\begin{align*}
 \EE[X] &= \sum_{n=0}^{+\infty} n . \mathbb P[X=n] =\sum_{n=0}^{+\infty} n . C_{n} p(pq)^n \\
 &=p\sum_{n=0}^{+\infty}  C_{n} (n+1) (pq)^{n}-p\sum_{n=0}^{+\infty} C_{n} (pq)^{n}\\
 &= p \, \frac{d}{dx} \left ( xC(x)\right ) \bigg\rvert_{x=pq} -p\, C(pq)\\
 &= \frac{p}{p-q}-1 =\frac{q}{p-q} \ \ .
\end{align*}

For the second type $(p,q)$-Catalan random variable we have

\begin{align*}
 \EE[X] &= \sum_{n=0}^{+\infty} n . \mathbb P[X=n] =\sum_{n=1}^{+\infty} n . C_{n-1} (pq)^n \\
 &=\sum_{m=0}^{+\infty}  C_{m} (m+1) (pq)^m = pq \, \frac{d}{dx} \left (x C(x)\right ) \bigg\rvert_{x=pq} \\
 &= \frac{pq}{p-q} \ \ .
\end{align*}
\end{proof}

The main geometric combinatorial property of Catalan numbers used in this article is the following 
well known enumeration (see section 9 of \cite{K08} p.259)

\begin{proposition}\label{geometric_property}
 The number of paths in $\NN^2$ going up and right at each step that start 
 at $(0,0)$ and end-up at $(n+1, n+1)$ without touching the first diagonal is $C_n$.
\end{proposition}

\section{Biased coin tossing.}\label{appendix_biased_coin_tossing}

The following lemma is useful for both LSM and EFSM strategies.

\begin{lemma}\label{col}
 Let $0<\gamma<1$ and $n \in \mathbb{N}$. Let $\omega \in \{ 0, 1 \}^{n}$ denotes 
 the outcome of tossing a biased coin $n$ times with
  $\mathbb{P} [\omega_i = 1] = \gamma$ for $i \in \{1,\ldots n \}$. 
  Let $Z (\omega) = \sup \{ i \in \{1,\ldots n \} ; \omega_i = 1 \} \cup \{ 0 \}$. 
  Then we have $\mathbb{E} [Z] = n+1 - \frac{1 - (1 - \gamma)^{n+1}}{\gamma}$.
\end{lemma}

\begin{proof}
  Note that $n+1 - Z = \inf \{ i\in\{ 1,\ldots, n\} ; \omega_{n+1-i} = 1\}\cup\{n+1\} = {\tilde Z}\wedge (n+1)$ 
  where ${\tilde Z}$ is the number of trials (stopping time) 
  before getting "Heads" when a coin is flipped repeatedly with a probability $\gamma$ 
  of getting "Heads" each time.  We have
\begin{align*}
  \mathbb{E} [\tilde{Z} \wedge (n+1)] & = \sum_{i = 1}^{n } i\mathbb{P}
  [\tilde{Z} = i] + (n+1)\mathbb{P} [\tilde{Z} > n ]\\
  & = \sum_{i = 1}^{n } i (\mathbb{P} [\tilde{Z} > i - 1] -\mathbb{P}
  [\tilde{Z} > i]) + (n+1)\mathbb{P} [\tilde{Z} > n ]\\
  & = \sum_{i = 0}^{n} \mathbb{P} [\tilde{Z} > i] = \sum_{i = 0}^{n}  
  (1 - \gamma)^i = \frac{1 - (1 - \gamma)^{n+1}}{\gamma}
\end{align*}
\end{proof}

\section{Poisson Games.}\label{appendix_Poisson_games}

\begin{theorem}\label{poiga}
  Let $N (t)$ (resp. $N' (t)$) be a Poisson process with parameter $\alpha$
  (resp. $\alpha'$). Let $\tau$ be the stopping time defined by $\tau = \inf
  \{ t \in \mathbb{R}_+ ; N (t) = N' (t) + 1 \}$. If $\alpha > \alpha'$ then
  $\tau \in L^1, N(\tau)\in L^1, \mathbb{E} [\tau] = \frac{1}{\alpha - \alpha'}$ 
  and $\mathbb{E} [N(\tau)] = \frac{\alpha}{\alpha - \alpha'}$.
\end{theorem}

\begin{proof}
 The proof is similar to the proof of Theorem 4.4 from \cite{GPM17}. 
 For $t \geq 0$, the stopping time $\tau \wedge t$ is bounded. So, by applying 
 Doob's Theorem (\cite{W}) to the martingales $N (t) -\alpha t$ and $N' (t) - \alpha' t$, we get
  \begin{align*}
    \alpha \mathbb{E} [\tau \wedge t]  &=  \mathbb{E} [N (\tau \wedge t)]\\
    & = \mathbb{E} [N (\tau \wedge t) | \tau \leq t] \mathbb{P}
    [\tau \leq t] +\mathbb{E} [N (\tau \wedge t) | \tau > t] \mathbb{P}
    [\tau > t]\\
    & = \mathbb{E} [N (\tau) | \tau \leq t] \mathbb{P} [\tau
    \leq t] +\mathbb{E} [N (t) | \tau > t] \mathbb{P} [\tau > t]\\
    & = \mathbb{E} [N' (\tau) + 1| \tau \leq t] \mathbb{P} [\tau
    \leq t] +\mathbb{E} [N (t)] \mathbb{P} [\tau > t]\\
    & = \mathbb{E} [N' (\tau) | \tau \leq t] \mathbb{P} [\tau
    \leq t] + \alpha t\mathbb{P} [\tau > t] +\mathbb{P} [\tau \leq
    t]\\
    & = \mathbb{E} [N' (\tau \wedge t) | \tau \leq t] \mathbb{P}
    [\tau \leq t] +\mathbb{E} [N' (\tau \wedge t) | \tau > t]
    \mathbb{P} [\tau > t] \\ 
    &  -\mathbb{E} [N' (\tau \wedge t) | \tau > t]
    \mathbb{P} [\tau > t] + \alpha t\mathbb{P} [\tau > t] +\mathbb{P} [\tau \leq t]\\
    & = \mathbb{E} [N' (\tau \wedge t)] -\mathbb{E} [N' (t) | \tau > t]
    \mathbb{P} [\tau > t] + \alpha t\mathbb{P} [\tau > t] +\mathbb{P} [\tau
    \leq t]\\
    & = \alpha' \mathbb{E} [\tau \wedge t] - \alpha' t\mathbb{P} [\tau >
    t] + \alpha t\mathbb{P} [\tau > t] +\mathbb{P} [\tau \leq t]
  \end{align*}
  So we have $(\alpha - \alpha')  (\mathbb{E} [\tau \wedge t] - t\mathbb{P} [\tau >t]) =\mathbb{P} [\tau \leq t]$ and
$$
(\alpha - \alpha') \mathbb{E} [\tau \wedge t| \tau \leq t] \mathbb{P} [\tau \leq t] = 
(\alpha - \alpha')  (\mathbb{E}  [\tau \wedge t] - t\mathbb{P} [\tau > t]) = \mathbb{P} [\tau \leq t] \ .
$$
  Therefore,
  $$
    \mathbb{E} [\tau \boldsymbol{1}_{\tau \leq t}] =  \frac{\mathbb{P}
    [\tau \leq t]}{\alpha - \alpha'} \ .
  $$
  Making $t \rightarrow \infty$, by monotone  convergence we have $\tau \in L^1$ and 
  $\mathbb{E} [\tau] = \frac{1}{\alpha - \alpha'}$. Moreover, using Doob's Theorem again, 
  we have for $t>0$,

  \begin{align*}
    \mathbb{E} [N(\tau) \boldsymbol{1}_{\tau \leq t}]
    & = \mathbb{E} [N (\tau \wedge t)] - \mathbb{E} [N (t) | \tau > t] 
    \mathbb{P} [\tau > t] \\
    &= \mathbb{E} [\tau \wedge t] - \alpha t\mathbb{P} [\tau > t] \\
    &= \alpha \mathbb{E} [\tau \boldsymbol{1}_{\tau \leq t}]
  \end{align*}
  So, using the monotone convergence theorem again, 
  we get $\mathbb{E} [N(\tau)] = \frac{\alpha}{\alpha - \alpha'}$
\end{proof}

\end{document}